\documentclass[acmsmall,screen]{acmart}
\setcopyright{none}
\copyrightyear{2025}
\acmYear{2025}
\acmDOI{}
\usepackage{fancyhdr}

\usepackage{amsthm}

\usepackage{bbold}

\usepackage{eucal} 
\usepackage[colorinlistoftodos,
]{todonotes}
\setuptodonotes{color=black!10,bordercolor=black!30,textcolor=black}

\usepackage{xcolor}

\usepackage[frozencache]{minted}
\setminted{
    fontfamily=tt, 
}
\definecolor{LightGray}{gray}{0.95}

\usepackage{tikz}
\usetikzlibrary{automata,positioning,arrows.meta,backgrounds,fit,calc}
\tikzset{arr/.style={-Latex}}
\usepackage{pgfmath}
\usepackage{tikz-cd}
\usepackage[a]{esvect} 
\usepackage{bm} 
\usepackage{tabularray}
\UseTblrLibrary{diagbox}

\newcommand{\cS}{\mathcal{S}}
\newcommand{\cT}{\mathcal{T}}
\newcommand{\nc}{\mathbb{n}} 
\DeclareMathOperator{\Ob}{Ob}
\DeclareMathOperator{\dom}{dom}
\DeclareMathOperator{\cod}{cod}

\newcommand{\ENDO}[1]{#1^{\bm{\circlearrowright}}}

\begin{document}
\title[Exploring Semigroupoids]{Computational Exploration of Finite Semigroupoids}
\author{Attila Egri-Nagy}
\affiliation{\institution{Akita International University}\city{Akita-city}\country{Japan}}\email{egri-nagy@aiu.ac.jp}
\author{Chrystopher L.~Nehaniv}
\affiliation{\institution{University of Waterloo}\city{Waterloo}\country{Canada}}
\email{chrystopher.nehaniv@uwaterloo.ca}
\settopmatter{printacmref=false}
\settopmatter{printfolios=true}
\renewcommand\footnotetextcopyrightpermission[1]{}
\pagestyle{fancy}
\fancyfoot{}
\fancyfoot[R]{miniKanren'25}
\fancypagestyle{firstfancy}{
 \fancyhead{}
 \fancyhead[R]{miniKanren'25}
 \fancyfoot{}
}
\makeatletter
\let\@authorsaddresses\@empty
\makeatother
\begin{abstract}
Recent algorithmic advances in algebraic automata theory drew attention to semigroupoids (semicategories).
These are mathematical descriptions of typed computational processes, but they have not been studied systematically in the context of automata.
Here, we use relational programming to explore finite semigroupoids to improve our mathematical intuition about these models of computation.
We implement declarative solutions for enumerating abstract semigroupoids (partial composition tables), finding homomorphisms, and constructing (minimal) transformation representations.
We show that associativity and consistent typing are different properties, distinguish between strict and more permissive homomorphisms, and systematically enumerate arrow-type semigroupoids (reified type structures).
\end{abstract}
\keywords{semigroupoid, semigroup, category, relational programming, morphisms, representations, \texttt{Clojure}, \texttt{core.logic}, \texttt{miniKanren}}

\maketitle
\thispagestyle{firstfancy}

\section{Introduction}
\emph{Computational exploration} is the use of computers for epistemological purposes.
We do not automate a well-understood and specified process.
Instead, we get to know something by writing and running the code.
This paper has the added twist of being a bit self-referential: we will study models of computation by this computational learning method.
The reader can benefit both from learning about algebraic automata theory, and from seeing a novel application of relational programming.

\subsection{The Transition from Semigroups to Semigroupoids in Algebraic Automata Theory}
\emph{Algebraic automata theory} studies finite state automata as algebraic structures, namely as \emph{semigroups} \cite{Howie95}.
These are sets equipped with an associative binary operation.
The connection between automata and algebra is easy to see: input symbols correspond to transformations (total functions) of the state set, thus their composition is associative.

The Krohn-Rhodes Prime Decomposition Theorem \cite{primedecomp65} is a seminal results of the field, stating that all automata can be built from simpler basic building blocks in a hierarchical way.
There are different ways of computing such decompositions.
The Covering Lemma \cite{FromRelToEmulationNehaniv1996} is a particularly simple method that produces a two level decomposition.
It has a recent software implementation \cite{egrinagy2024relation}, but there is one caveat: except the top level, the components are not necessarily semigroups.
Their composition is partially defined, thus they are \emph{semigroupoids}, i.e., categories with optional identity morphisms.
Objects of categories can be seen as constraints for composition, in other words, typing rules.
Therefore, introducing \emph{types} to semigroups is needed for hierarchical decomposition algorithms \cite{sgpoiddec}.

Beyond a general decomposition algorithm, semigroupoids have other benefits.
They allow more efficient representations of computations in terms of fewer elements.
When the operation is not defined, we are not forced to define it somehow.
This has security implications too.
Typed composition can restrict the possible computations the system can perform.
Typing constrains composability and thus, it limits the connection between the subsystems.
In contrast, if the objects are sufficiently connected by functions, their dynamics can be fully transferred (Example \ref{example:communicating_vessels}).

Despite some specialized interests in typed automata \cite{barr1995category}, semigroupoids were not studied systematically.
Therefore, this paper will do \emph{a computational exploration of semigroupoids}.
We discuss algorithms and their outputs to enrich our intuition about typed computation.

\subsection{Paper Structure}

In Section \ref{sect:mathcompprelim} we briefly terminology for search problems and introduce the software tools.
Section \ref{sect:abstractsgpoids} is about the abstract representation of semigroup(oid)s, i.e., enumerating composition tables.
Section \ref{section:morphisms} studies the homomorphisms of semigroupoids, with a particular interest in the homomorphic images that reify the type structure of a base semigroupoid.
Section \ref{sect:transrep} is about concrete transformation representations.

The paper separates the mathematical theory, the \emph{specification} level, and the language-specific \emph{implementation} details.
This distinction is shown in the subsection titles.
For reimplementation purposes, a purely mathematical reading is possible.

\section{Mathematical and Computational Preliminaries}
\label{sect:mathcompprelim}

We prefer to do the exploration in a declarative manner, by describing the properties of the objects.
This way the mathematical description is closer to the source code and verification is easier.
Mathematically, we describe the objects of study as constraint satisfaction problems.
On the implementation level, we move from \emph{functional} to \emph{relational} programming.

\subsection{Theory: Constraint Satisfaction Problems}

Here we fix a general definition to give a template and notation for the particular problems later.

\begin{definition}[Constraint Satisfaction Problem CSP]
A \emph{constraint satisfaction problem} is a triple $(V,D,C)$, where $V$ is a set of \emph{variables} $V=\{V_1,\ldots,V_n\}$, each taking values from a corresponding \emph{domain}  in $D=\{D_1,\ldots,D_n\}$, $v_i\in D_i$,  and $C=\{C_1,\ldots, C_m\}$ is a set of \emph{constraints} (relations) defined on the variables.
\end{definition}
A \emph{solution} is a particular assignment of values to the variables $(v_1,\ldots,v_n)$,  satisfying all constraints.
The search space size can easily be calculated by the product of the domain sizes, $\prod_{i=1}^{n}|D_i|$.
This is often a prohibitively large value; thus, the success of the search depends on the efficiency of the constraints, i.e., their capability of ruling out parts of the search tree.

\subsection{Implementation: \texttt{Clojure}, \texttt{core.logic}, and \texttt{kigen}}
The \texttt{kigen} \cite{kigen} computer algebra software package contains reimplementations of several algorithms from other packages \texttt{Semigroups} \cite{Semigroups}, \texttt{SmallSemi} \cite{smallsemi}, \texttt{SubSemi} \cite{subsemi}, and \texttt{SgpDec} \cite{SgpDec}, all written for the \texttt{GAP} computer algebra system \cite{GAP4}.
The main purpose of \texttt{kigen} is to provide a different implementation for verification purposes.
It also contains the implementation presented in this paper.

\texttt{GAP} has its own functional and object-oriented programming language built on top of a kernel written in \texttt{C}, while \texttt{kigen} is written in \texttt{Clojure} \cite{Clojure2020}, a language hosted on the \texttt{Java Virtual Machine}.
Agreeing results on different architectures and by different implementations increase our confidence that the computational output is not corrupted by some software artifact.
The complete enumeration \cite{T4enum} of all 4-state finite automata (degree 4 transformation semigroups) was successfully reproduced in \texttt{kigen}.

Yet another advantage of the language choice is the accessibility of relational programming.
The library \texttt{core.logic} implements the \texttt{miniKanren} system \cite{MiniKanren}.
We will refer to the library in the code as \texttt{l}.
The core unification algorithm is famously described on a single page in the book \cite{ReasonedSchemer}.
Since the logic engine is directly accessible from the language, a convenient workflow can be adapted for relational programming:
\begin{enumerate}
\item Work out the functional version \texttt{f} of a function.
\item Turn that function \texttt{f} into a goal \texttt{fo} by adding the output of \texttt{f} as the last argument.
\end{enumerate}
This ensures that by the time of writing the goal, which is more difficult as it is less familiar, the functional understanding of the problem is already complete.

The choice of relational programming as the programming paradigm for the project is not for maximizing execution speed, but for improving our understanding of the mathematical object.
Therefore, no attempt is made to deal with the internal parameters of the logic engine.
It is treated as a black box.
We put emphasis on the declarative definitions of semigroupoids, but still enjoy the benefits of the descriptions being executable.

\section{Enumerating Abstract Semigroupoids}
\label{sect:abstractsgpoids}

We start with semigroups and introduce composition tables as their abstract representations.
By constructing all associative composition tables, we can enumerate abstract semigroups.
Thus, we partially reproduce the enumeration in \cite{smallsemi}.

Semigroups are \emph{typeless}, or more precisely \emph{single-typed}.
Composition is free of constraints: we can always compose two semigroup elements.
When we include types, composition is subject to typing rules and consequently it becomes partial.
When enumerating abstract semigroupoids, we need to pay attention both to associativity and to the typing rules.
As categories are special semigroupoids, their enumeration \cite{2014CountingFiniteCats} is contained in the enumeration of semigroupoids.

\subsection{Theory: Semigroups and their Abstract Representations -- Composition Tables}

Semigroups can be defined directly, or relative to more familiar algebraic structures, namely groups, that are special cases of semigroups.
Here is the direct definition.

\begin{definition}[Semigroup]
A \emph{semigroup} is a set $S$ with an associative binary operation $S\times S\rightarrow S$.
\end{definition}
We call this binary operation \emph{composition}, but \emph{multiplication} is also a common name.
As in automata theory, we write composition from left to right, $st$ is $s$ then $t$, mimicking the sequential processing of input symbols, and the concatenation of words.
Unfortunately, this is the opposite of the current convention in category theory, though it was not always the case (e.g., \cite{lawveretheories}).

A \emph{monoid} is a semigroup with a (unique) identity element: $\exists e\in S$ such that $es=se=s, \forall s\in S$.
In category theory, monoids are single-object categories.
Requiring inverses to each element in a monoid yields a \emph{group}, which is a fundamental object in modern mathematics \cite{GroupsAndSymmetry1988}.
Semigroups can be viewed as generalizations of groups with partial or no symmetries.

In a semigroup, any two elements can be composed; thus, a data structure containing the rules of composition comprises a  full description.

\begin{definition}[Semigroup Composition Table]
The \emph{composition table} of a semigroupoid $S$ of $|S|=n$ elements is an $n\times n$ square array.
The rows and columns are indexed by the set of elements.
The $(a,b)$ entry is the composite $ab$.
\end{definition}

\subsection{Implementation: Finding Associative Composition Tables}
Composition tables are implemented as a vector of vectors.
Consequently, the composition operation requires two vector lookups, \texttt{nth} calls.
There is no directly implemented goal for \texttt{nth} in \texttt{core.logic}, therefore we need to write a recursive version for a general sequential collection.
\begin{minted}[bgcolor=LightGray]{clj}
(defn ntho
  "Succeeds if the sequential collection `coll` has `content` at the position
   `index`. Classic recursive implementation."
  ([coll index content] (ntho coll index content 0)) ;initializing counter
  ([coll index content i]
   (l/matche [coll]
     ([[head . tail]]
      (l/conde [(l/== head content) (l/== index i)] ;match at right place
               [(ntho tail index content (inc i))]))))) ;otherwise recur
\end{minted}
Note that the relational version is lot more forgiving than the functional \texttt{nth}.
Indices other than non-negative integers cause no issues.
Over-indexing is also not a problem, since the recursion naturally stops due to the matching failing for the tail at the end of the list.
Thus, we can implement a goal for composition by literally translating the functional code for composition into relational form.
\begin{minted}[bgcolor=LightGray]{clj}
(defn composo
  "Succeeds if `a` composed with `b` is `ab` in the composition table `S`."
  [S a b ab]
  (l/fresh
   [row]
   (ntho S a row) ;the row of a
   (ntho row b ab))) ;the bth entry in that row should be ab
\end{minted}
We can run queries to ask what is the result of composition, what pairs lead to a particular result, and we can also build suitable composition tables.
\begin{minted}[bgcolor=LightGray]{clj}
(def FF [[0 1 2] ;the flip-flop semigroup, famous in automata theory
         [1 1 2]
         [2 1 2]])
(l/run* [q] (composo FF 0 1 q)) ; What is 01 in the flip-flop?
(1)
(l/run* [a b] (composo FF a b 1)) ; finding all pairs composing to 1
([0 1] [1 0] [1 1] [2 1])
(l/run* [q r s t] (composo [[q r] [s t]] 0 1 1)) ; all tables with 01=1
([_0 1 _1 _2])
\end{minted}
Finally, we define the goal for associativity for a given triple of arrows.
We need to check that $(ab)c=a(bc)$.
Another straight functional to relational translation.
\begin{minted}[bgcolor=LightGray]{clj}
(defn associativo
  "The goal for associativity for a given triple."
  [S a b c]
  (l/fresh
   [ab bc abc]
   (composo S ab c abc) ; (ab)c=abc
   (composo S a bc abc) ; a(bc)=abc
   (composo S a b ab)
   (composo S b c bc)))
\end{minted}
One \texttt{composo} goal can be eliminated for triples with identical elements.
We can also apply the technique used for solving sudoku \cite{2014joyclojure}, and start with a partially filled composition table.

Without any optimization, this general search can find all finite semigroups up to size 4, which is far from the state-of-the art size 10 \cite{smallsemi}.
The direct enumeration of semigroups is not too useful.
Asymptotically, most finite semigroups are 3-nilpotent, meaning that the semigroup has a zero element, and multiplying any three elements gives that zero.

\subsection{Theory: Introducing Types}

Adding types to semigroups can be approached from several directions.
One way is the `oidification', a way of categorification, where we add more objects in addition to the single object of an algebraic structure.
This suggests the name \emph{semigroupoid}.
It can be viewed as a slightly defective, non-unital category, where the identity arrows are not guaranteed to exist.
Thus, an alternative name is \emph{semicategory}.
An elementary understanding might start from a directed graph with the special property that every directed path has a single arrow ``equal to it'', the result of composing of the arrows along the path.
Thus, a semigroupoid is a transitively closed directed graph.
From a programming perspective, these arrows may be viewed as functions with the specified input and output types.
We will use the name semigroupoid.
Here is the precise definition:
\begin{definition}
  \label{def:sgpoid}
  A \emph{semigroupoid} consists of \emph{objects}, \emph{arrows} between the objects, and an \emph{associative} \emph{composition} operation.
For a semigroupoid $\cS$ we use the following notation.
\begin{description}
  \item[objects] The set of objects is denoted by $\Ob(\cS)$. For objects we use letters $\sigma,\tau,\rho,\kappa$. We also call these (formal) \emph{types}.
  \item[arrows] The total set of arrows is seldom mentioned. Instead, for an object pair $(\sigma,\tau)$ we talk about the set of arrows between those objects, denoted by $\cS(\sigma,\tau)$. It can be empty, or contain one or more arrows. An arrow $a\in\cS(\sigma,\tau)$ has \emph{domain} $\dom(a)=\sigma$, and \emph{codomain} $\cod(a)=\tau$ (also called \emph{source} and \emph{target}). Alternatively, we can write $a:(\sigma,\tau)$, which reads as `$a$ has type  $(\sigma,\tau)$'.
  \item[composition] Two arrows $a$ and $b$ are \emph{composable} if $\cod(a)=\dom(b)$.
    The composition is denoted by concatenation $(a,b) \mapsto ab$,
    which type-chekcs as  $a:(\sigma,\tau)$, $b:(\tau,\rho)$, and $ab:(\sigma,\rho)$.
    Here is a diagram for composition.
\begin{center}\begin{tikzcd}
  \sigma \arrow[r,"a",arr] \arrow[rr,"ab"',bend right=15,arr] & \tau \arrow[r,"b",arr] & \rho
\end{tikzcd}
\end{center}
 \item[associativity] Composition should satisfy the \emph{associativity} condition: $a(bc)=(ab)c$. Consequently, the composite $abc$ is a well-defined arrow of type $(\dom(a),\cod(c))$.
\end{description}
\end{definition}
The semigroupoid structure describes when arrows can be composed and the result of the composition.
Like for semigroups, we can use composition tables, but the entries for non-composable pairs of arrows are undefined.
When the arrows are grouped by type, the table has rectangular blocks of undefined cells.

\begin{example}[Two-type semigroupoid \cite{sgpoiddec}] Fig.~\ref{2obj6arr} shows a semigroupoid with two types $\sigma,\tau$. Morphisms $a,b$ are of type $\ENDO{\sigma}$, $c,d,e$ are of type $\vv{\sigma\tau}$, and $f$ is of type $\ENDO{\tau}$.
  \label{ex:2obj-sgpoid}
\end{example}

\begin{figure}[ht]
\begin{center}
\begin{tblr}{Q[c,m]Q[c,m]Q[c,m]} 
\begin{tikzpicture}[shorten >=1pt, node distance=2cm, on grid, auto,inner sep=2pt]
    \node[draw, circle] (X)   {$\sigma$};
    \node[draw, circle] (Y) [right of=X] {$\tau$};
    \path[->,every loop/.append style=-{Latex}]
    (X) edge [arr, loop above] node {$a$} (X)
    (X) edge [arr, loop below] node {$b$} (X)
    (X) edge [arr] node {$d$} (Y)
    (X) edge [arr,bend left=42, above] node {$c$} (Y)
    (X) edge [arr,bend right=42, below] node {$e$} (Y)
    (Y) edge [arr,loop right] node {$f$} (Y);
\end{tikzpicture}
&
\begin{tblr}{
  hline{2-8}={2-8}{0.4pt,solid},
  vline{2-8}={2-8}{0.4pt,solid}}
    & $a$ & $b$ & $c$ & $d$ & $e$ & $f$ \\
$a$ & $a$ & $b$ & $c$ & $d$ & $e$ & \\
$b$ & $b$ & $a$ & $c$ & $e$ & $d$ & \\
$c$ &  &  &  &  &  &  $c$\\
$d$ &  &  &  &  &  &  $c$\\
$e$ &  &  &  &  &  &  $c$\\
$f$ &  &  &  &  &  &  $f$ \\
\end{tblr}
&
\begin{tblr}{
hline{2-5}={2-4}{0.4pt,solid},
vline{2-5}={2-4}{0.4pt,solid}}
           & $\ENDO{\sigma}$ & $\vv{\sigma\tau}$ & $\ENDO{\tau}$  \\
$\ENDO{\sigma}$ & $\ENDO{\sigma}$ & $\vv{\sigma\tau}$ &  \\
$\vv{\sigma\tau}$ &  &  & $\vv{\sigma\tau}$ \\
$\ENDO{\tau}$ &  & & $\ENDO{\tau}$ \\
\end{tblr}
\end{tblr}
\end{center}
\caption{A semigroupoid with two objects and six arrows. Diagram of objects and arrows on the left, the corresponding composition table in the middle, and the simplified composition table only with the arrow types on the right.}
\label{2obj6arr}
\end{figure}
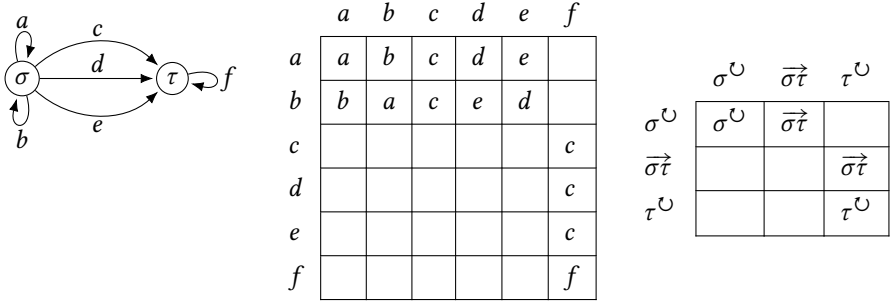

The above definition of a semigroupoid (Definition \ref{def:sgpoid}), as all the standard definitions, is a bit vague about the associativity condition.
When both pairs in the triple are composable, associativity can be checked without any problem.
What happens when one or both pairs fail to compose?
This is a crucial question when enumerating abstract semigroupoids.
We need to decide whether a given table is associative or not.
When working with concrete representations, e.g., sets and transformations, this is not an issue since we know that function composition is associative.

\begin{example}
Figure~\ref{fig:3x3assocfail} shows two composition tables with three arrows.
Both are sparsely populated, so they do not have enough composition to have a triple of elements not leading to undefined result.
Still, we need to tell them apart, as one of them is actually a semigroupoid.
That's why it is crucial to define associativity properly for non-composable pairs.
\end{example}
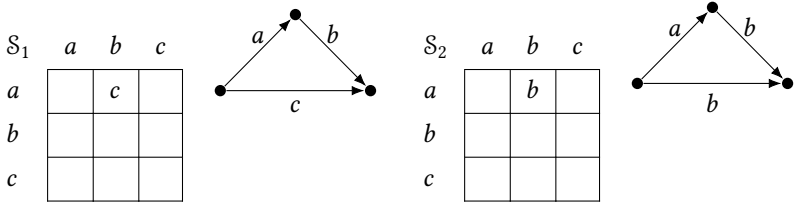
\begin{figure}
\begin{center}
\begin{tblr}{Q[c,m]Q[c,m]Q[c,m]Q[c,m]} 
\begin{tblr}{
  hline{2-5}={2-5}{0.4pt,solid},
  vline{2-5}={2-5}{0.4pt,solid}}
 $\cS_1$   & $a$ & $b$ & $c$  \\
$a$ &  & $c$ &   \\
$b$ &  &     &   \\
$c$ &  &  &  \\
\end{tblr}
&
\begin{tikzpicture}
  \tikzset{obj/.style={fill,circle,inner sep=1.5pt},
          arr/.style={-Latex}}
  \node at (0,0) [obj] (t0) {};
  \node at (1,1) [obj] (t1) {};
  \node at (2,0) [obj] (t2) {};

  \path[->,every loop/.append style=arr]
  (t0) edge [arr] node [above] {$a$} (t1)
  (t1) edge [arr] node [above] {$b$} (t2)
  (t0) edge [arr] node [below] {$c$} (t2);
\end{tikzpicture}
&
\begin{tblr}{
  hline{2-5}={2-5}{0.4pt,solid},
  vline{2-5}={2-5}{0.4pt,solid}}
 $\cS_2$   & $a$ & $b$ & $c$  \\
$a$ &  & $b$ &   \\
$b$ &  &     &   \\
$c$ &  &  &  \\
\end{tblr}
&
\begin{tikzpicture}
  \tikzset{obj/.style={fill,circle,inner sep=1.5pt},
          arr/.style={-Latex}}
  \node at (0,0) [obj] (t0) {};
  \node at (1,1) [obj] (t1) {};
  \node at (2,0) [obj] (t2) {};

  \path[->,every loop/.append style=arr]
  (t0) edge [arr] node [above] {$a$} (t1)
  (t1) edge [arr] node [above] {$b$} (t2)
  (t0) edge [arr] node [below] {$b$} (t2);
\end{tikzpicture}

\end{tblr}
\end{center}
\caption{Two $3\times 3$ composition tables. In $\cS_1$, $ab=c$ is the only composition and the diagram shows that it is a semigroupoid. However, $\cS_2$ fails associativity for $aab$, since $(aa)b=\nc b=\nc$, where $\nc$ stands for `not composable', while $a(ab)=ab=b$, and $b\neq \nc$. The diagram also shows that it is not a semigroupoid. We could make $\cS_2$ into a semigroupoid by making $a$ idempotent, $aa=a$, but in $\cS_2$ $a$ is not composable with itself.}
\label{fig:3x3assocfail}
\end{figure}

We represent the result of the composition of arrows that are not composable by a new value $\nc$, meaning non-defined, or \emph{not composable}.
To some extent, we treat $\nc$ as if it were an arrow: composing any arrows with $\nc$ yields $\nc$.
Most importantly, when checking associativity, we treat it as any other value.
This way we can catch situations when $(ab)c$ yields $\nc$, but $a(bc)$ is an arrow, violating associativity.
Of course, ultimately, $\nc$ is not an arrow.
It can appear in the cells of the composition table, but there is no row or column for $\nc$.

Now we can be more precise: the associativity condition is satisfied if
\begin{enumerate}
\item $a(bc)=(ab)c$, when $ab$ and $bc$ are both defined; or,
\item $ab=\nc \implies a(bc)=\nc$; or,
\item $bc=\nc \implies (ab)c=\nc$; or,
\item $bc=ab=\nc$.
\end{enumerate}

\subsection{Implementation: Dealing with Non-composable Pairs}

The main mechanism for enumerating abstract semigroupoids is the same as for enumerating abstract semigroups.
We check associativity for all triples.
However, we have to be careful when  composition is not defined.
We represent the lack of composability by the keyword \texttt{:n}, meaning `not composable' (same as $\nc$).
The goal for associativity in a semigroupoid is just simply  writing down all the cases that can happen.
\footnote{The road to this solution was not as simple though.
In the functional code, \texttt{nil} is a convenient choice for representing non-composable pairs in the table, as it works well with \texttt{when} and \texttt{when-not}. The \texttt{nil} value is also conspicuous when printing the composition table.
However, it is difficult to handle \texttt{nil} values in the logic engine at the same time as working with values from a finite domain.
Also, it interferes with logic variables without values yet.
The former error is easy to catch since it throws an exception, but the latter was caught only by inconsistencies in the mathematical results.}

\begin{minted}[bgcolor=LightGray]{clj}
(defn sgpoid-associativo
  "The goal for associativity for a given triple in a semigroupoid."
  [S [a b c]]
  (l/fresh
   [ab bc]
   (composo S a b ab)
   (composo S b c bc)
   (l/conda
    [(l/conde
      [(l/== ab :n) (composo S a bc :n)]
      [(l/== bc :n) (composo S ab c :n)]
      [(l/== bc :n) (l/== ab :n)])]
    [(l/fresh
      [abc]
      (composo S ab c abc)
      (composo S a bc abc))])))
\end{minted}
Unlike the semigroup case, enumerating composition tables that satisfy the generalized associativity condition does not necessarily give semigroupoids.
We need to check whether the typing rules make sense or not.

\subsection{Theory: Inferencing Type Structure}

Associativity alone does not guarantee that a composition table is a semigroupoid.
It also needs a consistent \emph{type structure}: the number of types (objects) and the domain-codomain pairs for arrows.
This problem is not the same as \emph{type inference} in programming, where we want to find the type of an expression.
When two semigroupoids have identical type structures, we say they are \emph{type compatible}.

Having a type structure and being associative are different properties.
Simple examples in Figure~\ref{fig:typednonassoc} demonstrate that a table can have consistent types and fail the associativity condition.
We can also have an associative table with inconsistent typing.

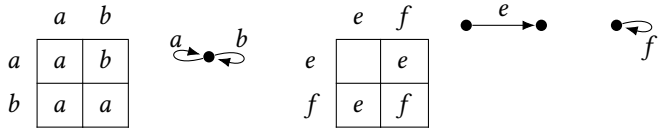
\begin{figure}
\begin{center}
\begin{tblr}{Q[c,m]Q[c,m]Q[c,m]Q[c,m]} 
\begin{tblr}{
  hline{2-4}={2-4}{0.4pt,solid},
  vline{2-4}={2-4}{0.4pt,solid}}
    & $a$ & $b$   \\
$a$ & $a$ & $b$    \\
$b$ & $a$ &  $a$      \\
\end{tblr}
&
\begin{tikzpicture}
  \tikzset{obj/.style={fill,circle,inner sep=1.5pt},
          arr/.style={-Latex}}
  \node at (0,0) [obj] (t0) {};
  \path[->,every loop/.append style=arr]
  (t0) edge [arr,loop left] node [above] {$a$} (t0)
  (t0) edge [arr,loop right] node [above] {$b$} (t1);
\end{tikzpicture}
&
\begin{tblr}{
  hline{2-4}={2-4}{0.4pt,solid},
  vline{2-4}={2-4}{0.4pt,solid}}
    & $e$ & $f$   \\
$e$ &  & $e$    \\
$f$ & $e$ &  $f$      \\
\end{tblr}
&
\begin{tikzpicture}
  \tikzset{obj/.style={fill,circle,inner sep=1.5pt},
          arr/.style={-Latex}}
  \node at (0,0) [obj] (t0) {};
  \node at (1,0) [obj] (t1) {};
  \node at (2,0) [obj] (t2) {};
  \path[->,every loop/.append style=arr]
  (t0) edge [arr] node [above] {$e$} (t1)
  (t2) edge [arr,loop right] node [below] {$f$} (t2);
\end{tikzpicture}
\end{tblr}
\end{center}
\caption{For the left table, type structure inference is trivial for the composition table, since everything is composable. Thus we have a semigroup candidate. However, associativity fails: $(bb)b=ab=b$, but $b(bb)=ba=a$. Type structure does not imply associativity. On the right, the table is associative, but we cannot produce consistent typing rules. Arrow $e$ is not self-composable, thus $\dom(e)\neq\cod(e)$. For arrow $f$ we have $\dom(f)=\cod(f)$. However, the table shows that we should be able to compute both $fe$ and $ef$, i.e., $\cod(f)=\dom(e)$ and $\cod(e)=\dom(f)$, implying $\dom(e)=\cod(e)$, contradiction. Therefore, associativity does not imply consistent type structure.}
\label{fig:typednonassoc}
\end{figure}

Constraints for type inferencing are easy to describe by checking equality for domains and codomains.
For each arrow $a$ we have two logic variables, one for the domain $\dom(a)$, the other one for its codomain $\cod(a)$.
Given that we try to have $m$ types, each variable can take $m$ values.
If the table has $n$ arrows, the search space size is $m^{2n}$.
The constraints are
\begin{enumerate}
  \item when $a$ is not composable with $b$, then $\cod(a)\neq\dom(b)$;
  \item when $ab$ is a valid composition (including the case $aa$), we need $\cod(a)=\dom(b)$;
  \item when $ab=c$, then $\dom(a)=\dom(c)$ and $\cod(b)=\cod(c)$.
\end{enumerate}
The first two cases already give $n^2$ constraints, and depending on the number of composable pairs, we may get another $n^2$.
When we have $2n^2$ constraints for a complete composition table, then we have a semigroup and consequently only a single type is admissible.
The other extreme is an empty table (see Figure~\ref{fig:emptycomptab}), where we have $n^2$ statements about domain-codomain pairs not being equal, still allowing diversity in type structure.
We can route non-composable arrows arbitrarily, as long as they do not become composable.

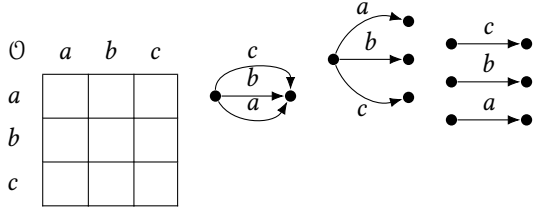
\begin{figure}
\begin{center}
\begin{tblr}{Q[c,m]Q[c,m]Q[c,m]Q[c,m]} 
\begin{tblr}{
  hline{2-5}={2-5}{0.4pt,solid},
  vline{2-5}={2-5}{0.4pt,solid}}
 $\mathcal{O}$   & $a$ & $b$ & $c$  \\
$a$ &  &  &   \\
$b$ &  &     &   \\
$c$ &  &  &  \\
\end{tblr}
&
\begin{tikzpicture}
  \tikzset{obj/.style={fill,circle,inner sep=1.5pt},
          arr/.style={-Latex}}
  \node at (0,0) [obj] (t0) {};
  \node at (1,0) [obj] (t1) {};
  \path[->,every loop/.append style=arr]
  (t0) edge [arr,bend right=60] node [above] {$a$} (t1)
  (t0) edge [arr] node [above] {$b$} (t1)
  (t0) edge [arr,bend left=90] node [above] {$c$} (t1);
\end{tikzpicture}
&
\begin{tikzpicture}
  \tikzset{obj/.style={fill,circle,inner sep=1.5pt},
          arr/.style={-Latex}}
  \node at (0,0) [obj] (t0) {};
  \node at (1,0.5) [obj] (t1) {};
  \node at (1,0) [obj] (t2) {};
  \node at (1,-0.5) [obj] (t3) {};
  \path[->,every loop/.append style=arr]
  (t0) edge [arr,bend left=42] node [above] {$a$} (t1)
  (t0) edge [arr] node [above] {$b$} (t2)
  (t0) edge [arr,bend right=42] node [below] {$c$} (t3);
\end{tikzpicture}
&
\begin{tikzpicture}
  \tikzset{obj/.style={fill,circle,inner sep=1.5pt},
          arr/.style={-Latex}}
  \node at (0,0) [obj] (t0) {};
  \node at (1,0) [obj] (t1) {};
  \node at (0,.5) [obj] (t2) {};
  \node at (1,.5) [obj] (t3) {};
  \node at (0,-.5) [obj] (t4) {};
  \node at (1,-.5) [obj] (t5) {};
  \path[->,every loop/.append style=arr]
  (t0) edge [arr] node [above] {$b$} (t1)
  (t2) edge [arr] node [above] {$c$} (t3)
  (t4) edge [arr] node [above] {$a$} (t5);
\end{tikzpicture}
\end{tblr}
\end{center}
\caption{Semigroupoid $\mathcal{O}$ with an empty composition table. The emptiness of the table does not imply the absence of arrows, only the lack of composition. Indeed, $\mathcal{O}$ can have several different type structures (there are more), as the only constraint is that no arrows can be composed. The rightmost diagram is maximal in the sense that cannot be  more objects involved in forming domains and codomains.}
\label{fig:emptycomptab}
\end{figure}

\subsection{Implementation: Search for Minimal Number of Objects}
The constraints are implemented by unification, \texttt{l/==} and by \texttt{l/distincto}.
More interesting than the straightforward coding of the constraints is the problem of finding a minimal number of types needed for a given semigroupoid.
This is done by a simple meta-search algorithm.
Starting from $m=1$, we attempt to construct a type structure.
If it fails, we go to $m+1$.
We can stop the search at $m=2n$ for $n$ arrows.
In that case, each arrow gets its own distinct domain-codomain pairs.
Beyond this we would have objects not assigned to any arrows (see Figure~\ref{fig:emptycomptab}).
If we do not get a consistent type structure for any of these values, we conclude that the composition table is not a semigroupoid.

\section{Enumerating Morphisms of Semigroupoids}
\label{section:morphisms}

Homomorphisms are structure preserving functions and the structure of a semigroupoid is defined by arrow composition.
For finding homomorphisms, we consider semigroupoids as single-sorted algebras (arrows only), instead of two-sorted (objects and arrows).
The objects do not appear in the search algorithm, there are no logic variables for them.
However, mappings on arrows induce mappings on objects, thus the homomorphisms can be seen as semigroupoid functors.

Given semigroupoids $\cS$ and $\cT$, what are the homomorphisms $\varphi:\cS\rightarrow \cT$?
They are functions that take a composable pair of arrows in the source to a composable pair in the target.
What happens to non-composable pairs?
We have two choices.
We can be strict, requiring that non-composable pairs go to non-composable pairs.
This way we get homomorphisms that preserve type structures.
This could be the right choice when modeling computer security protocols.
Alternatively, we can allow non-composable pairs to become composable by checking compatibility only for composable pairs.
In algebra, a map collapsing everything into an identity arrow is a valid homomorphism.

\subsection{Theory: Homomorphisms}

A homomorphism $\varphi: \cS\rightarrow \cT$ sends each arrow $a\in\cS$ to an arrow $\varphi(a)\in\cT$, its \emph{image}.
This map must satisfy the condition of \emph{compatible operations}:
$$\varphi(ab)=\varphi(a)\varphi(b),\ \forall a,b\in\cS.$$
Compatible composition means that it does not matter whether we first compose in $\cS$ and then send the result to $\cT$, or first send the arrows first to their corresponding images in $\cT$, and then do the composition there.
For expressing the compatibility conditions as constraints, we turn the composition function in $\cS$ into a relation.

\begin{definition}[Composition function and relation]
The \emph{composition function} for semigroupoid $\cS$ is $c_\cS:\cS\times\cS\rightarrow\cS$, taking an ordered pair of arrows to their composite arrow: $(a,b)\mapsto ab$.
In general, turning around this function yields a relation from $\cS$ to $\cS\times\cS$, given by
$c_\cS^{-1}(d)=\{(a,b)\mid c_\cS(a,b)=d\}$.
In other words, $a$ and $b$ are $d$-related if $ab=d.$
Non-composable pairs are $\nc$-related.
\end{definition}

For deciding whether the compatibility conditions hold or not, the basic idea is that \emph{we take the composition relation of $\cS$ and substitute the arrows of $\cS$ with variables taking values from the arrows of $\cT$}.
Formally, $c_\cS^{-1}[a\mapsto\varphi(a)]$, turning the relation $d=ab$ in $\cS$ into the relation $\varphi(d)=\varphi(a)\varphi(b)$ in $\cT$.
The substitution turns the image $\varphi(a)$ into the variable $v_a$.

\begin{definition}[Semigroupoid Homomorphism CSP] Given $\cS=\{a_1,\ldots,a_n\}$ and $\cT$, the search problem for finding homomorphisms $\varphi:\cS\rightarrow \cT$ is the triple
  $$H_{\cS\rightarrow\cT}=\left(\{v_{0},\ldots,v_{n}\}, \cT, c_\cS^{-1}[a_i\mapsto v_i]\right),$$
  where $v_i$ is interpreted as $\varphi(a_i)$.
\end{definition}

When constructing $\varphi$, we need a logic variable for each element of the semigroupoid to determine its image.
Each variable can have any arrow from $\cT$ as its value.
Thus, search space size is $|\cT|^{|\cS|}$, which is exactly the number of functions from $\cS$ to $\cT$.

For semigroups, a bijective homomorphism is an \emph{isomorphism}.
In that case, we add the constraint that $v_{a_i}\neq v_{a_j}$ whenever $i\neq j$.
However, for semigroupoids, bijective homomorphisms need not be isomorphisms.
It depends on whether we include the $\nc$-relation in the constraints, or not.
A homomorphism is \emph{strict} if it satisfies compatibility conditions for $\nc$-related elements as well.
The situation is similar to graph theory.
Graph isomorphisms are defined with an if-and-only-if condition \cite{2020_GI_review}.
For two graphs to be isomorphic, one graph should have an edge whenever the other graph has one between the corresponding nodes.
There should not be an edge when the other one does not have it.
The map preserves adjacency, and non-adjacency.
Graph homomorphisms, on the other hand are defined by implication: edges in the source imply edges in the target, but we do not say anything about  non-edges.
In category theory, we distinguish between \emph{preservation} (if the source has the structure, then target also), and \emph{reflection} (if the image in the target has the structure, then the source has it too).

Similarly, in semigroupoids are composition is not defined for  all pairs of arrows.
When we transfer structure, the question is, do we care only about composable pairs, or all arrow pairs?
Shall the mapping only preserve, or reflect composition?
It may depend on the application.
If we care about the type structure, then we need to use \emph{strict homomorphisms} that reflect composition.
\begin{proposition}
If $\varphi:\cS\to \cT$ is a strict homomorphism, then $\cS$ and $\cT$ are type-compatible.
A bijective strict homomorphism is an isomorphism.
\end{proposition}

\begin{figure}
  \begin{center}
\begin{tblr}{Q[c,m]Q[c,m]Q[c,m]Q[c,m]} 
\begin{tblr}{
  hline{2-4}={2-4}{0.4pt,solid},
  vline{2-4}={2-4}{0.4pt,solid}}
  $\cS$  & $a$ & $b$   \\
$a$ & $a$ &     \\
$b$ &  &        \\
\end{tblr}
&
\begin{tikzpicture}
  \tikzset{obj/.style={fill,circle,inner sep=1.5pt},
          arr/.style={-Latex}}
  \node at (0,0) [obj] (t0) {};
  \node at (1,0) [obj] (t1) {};
  \node at (2,0) [obj] (t2) {};
  \path[->,every loop/.append style=arr]
  (t0) edge [arr,loop left] node [above] {$a$} (t0)
  (t2) edge [arr] node [above] {$b$} (t1);
\end{tikzpicture}
&
\begin{tblr}{
  hline{2-4}={2-4}{0.4pt,solid},
  vline{2-4}={2-4}{0.4pt,solid}}
  $\cT$  & $e$ & $f$   \\
$e$ & e & $f$    \\
$f$ & $f$ &  $e$      \\
\end{tblr}
&
\begin{tikzpicture}
  \tikzset{obj/.style={fill,circle,inner sep=1.5pt},
          arr/.style={-Latex}}
  \node at (0,0) [obj] (t0) {};
  \path[->,every loop/.append style=arr]
  (t0) edge [arr,loop right] node [above] {$e$} (t0)
  (t0) edge [arr,loop left] node [below] {$f$} (t0);
\end{tikzpicture}
\end{tblr}
\end{center}
\caption{A bijective homomorphism $\varphi:\cS\to\cT$ is given by $a\mapsto e$, since they are self-composable identity arrows,  and $b\mapsto f$ is forced by the bijectivity. However, finding such a homomorphism backwards is not possible: $e\mapsto a$, but then by bijectivity we would be forced to have $f\mapsto b$, thus the compatibility condition would be violated.}
\label{fig:asymmetricisomorphism}
\end{figure}
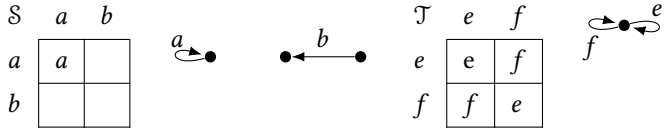

Alternatively, we can postulate that only the compatibility of composable pairs matter.
Arrows not composable in $\cS$ may be sent to arrows that are composable in $\cT$.
There are more of these more permissive morphisms, and we can recreate the (composable) dynamics of $\cS$ in more target semigroupoids $\cT$.
However, there are some strange consequences.
First, we can find bijective isomorphisms into a target structure, which is not even a semigroupoid.
It can fail associativity for the arrow pairs not composable in $\cS$.
Second, unlike isomorphisms, bijective homomorphism is not an equivalence relation by failing symmetry.
Figure~\ref{fig:asymmetricisomorphism} shows a pair of semigroupoids, one has a bijective homomorphism to the other one, but not vice versa.

\subsection{Implementation: Constructing Homomorphisms}

The call to the logic engine is a literal translation of the setup described in the mathematical part.
\begin{minted}[bgcolor=LightGray]{clj}
(defn morphism-search
  "Logic search for all homomorphisms of semigroupoid `S` to `T` given as
   composition tables. If `bijective?`, then only isomorphisms are enumerated.
   When `strict?`, non-composable arrow pairs cannot become composable."
  [S T bijective? strict?]
  (let [n (count S)
        phi (lvar-vector n)
        elts (range (count T))
        constraints (if strict?
                      (substitute (composition-relation S) phi)
                      (substitute (dissoc (composition-relation S) :n) phi))]
    (l/run*
     [q]
     (l/== q phi)
     (l/everyg (fn [elt] (l/membero elt elts)) phi)
     (if bijective?
       (l/distincto phi)
       l/succeed)
     (l/everyg (fn [[ab pairs]]
                 (l/everyg (fn [[a b]]
                             (composo T a b ab))
                           pairs))
               constraints))))
\end{minted}
There are 6 strict endomorphisms of the semigroupoid from Example \ref{ex:2obj-sgpoid}, and 9 endomorphisms in total. The extra three morphisms consist of sending everything to one identity arrow or the other, and one that sends everything to the identity in the first type, but all the inter-type arrows go to the non-identity arrow of the first type.

\subsection{Theory: Arrow-Type Semigroupoid}

The enumeration of semigroupoids drew attention to the consistency of the type structures: the connections between the types should not contradict the abstract composition rules.
Now, we pin down the minimal structure describing the connectivity of types, and it happens to be a semigroupoid.
\begin{definition}[Arrow-type semigroupoid]
Given a semigroupoid $\cS$, the \emph{arrow-type semigroupoid} $\cS_\rightarrow$ has the arrow-types (domain-codomain object pairs) of $\cS$ as its objects.
There is an arrow in $\cS_\rightarrow$ when the corresponding arrow-types are composable in $\cS$.
\end{definition}
There can be only one arrow between two objects in one direction, but there can be another one in the reversed direction; thus, these semigroupoids are not posets (they also potentially lack identities).

From the perspective of the decomposition algorithms that require a surjective homomorphism as a starting point (e.g., \cite{sgpoiddec}), the arrow-type semigroupoids can be a useful source of those first approximations.
We can always find a homomorphic image by `lumping together' arrows of the same type.
It also serves as an isomorphism invariant.

\begin{proposition}
For a semigroupoid $\cS$ with an associated type structure there exists a unique arrow-type semigroupoid $\cS_\rightarrow$.
\end{proposition}
\begin{proof}
  Send all arrows $a$ to the pair $(\dom(a),\cod(a))$.
  Composition is defined formally as $(x,y)(y,z)=(x,z)$.
Checking associativity is routine: $((x,y)(y,z))(z,w)=(x,z)(z,w)=(x,w)$ and $(x,y)((y,z)(z,w))=(x,y)(y,w)=(x,w)$.
\end{proof}

\begin{table}
  \begin{tblr}{colspec={Q[r]Q[r]Q[r]Q[r]Q[r]Q[r]Q[r]Q[r]Q[r]Q[r]Q[r]Q[r]Q[r]Q[r]Q[r]Q[r]},
  hline{3-28}={2-15}{0.2pt,solid},
  vline{2-16}={3-27}{0.2pt,solid},
  vline{7}={27}{0.2pt,solid}}
  & $\bullet$ \\
  $\rightarrow$  & $1$ & $2$ & $3$ & $4$ & $5$ & $6$ & $7$ & $8$ & $9$ & $10$ & 11 & 12 & 13 & 14  & $\sum$ \\
$1$ & $1$ & $1$ &  &  &  & & & & & & & & & &  2\\
$2$ &  & $3$ & $3$ & $1$ &  & & & & & & & & & &   7\\
$3$ &  & 1 & 8 & 8 & 3  & 1  & & & & & & & & &  21\\
$4$ &  & 1 & 8 & 23 & 23 &  11 & 3 & 1 & & & & & & &  70\\
$5$ &  &  & 6 & 34 & 67 & 64  & 32 & 11 & 3 & 1 & & & & &  218\\
6 & &  & 3 & 42 & 132 & 211 & 185 & 97 & 36 & 11 & 3 & 1  & & &  721\\
7 & & & 2 & 35 & 205 & 486 &  652 & 536 & 283 & 110 & 36 & 11  & 3 & 1 & 2360\\
8 & & &  & 27 & 254 & 925 &  1763 & 2063 & 1583 & 837 & 334 & 115 & \SetCell{gray!10} & \SetCell{gray!10} &  \\
9 & & & 1 & 14 & 260 &  1436 & 3905  & \SetCell{gray!10} & \SetCell{gray!10} & \SetCell{gray!10} & \SetCell{gray!10} & \SetCell{gray!10} & \SetCell{gray!10} & \SetCell{gray!10} \\
10 &   &  &  & 10 & 226 &  1908 & 7306  & \SetCell{gray!10} & \SetCell{gray!10} & \SetCell{gray!10} & \SetCell{gray!10} & \SetCell{gray!10} & \SetCell{gray!10} & \SetCell{gray!10} \\
 11 &  &  &  & 3 & 179 &  2182 & 11768  & \SetCell{gray!10} & \SetCell{gray!10} & \SetCell{gray!10} & \SetCell{gray!10} & \SetCell{gray!10} & \SetCell{gray!10} & \SetCell{gray!10} \\
 12   &  &  &  & 3 & 120 &  2246 & 16702  & \SetCell{gray!10} & \SetCell{gray!10} & \SetCell{gray!10} & \SetCell{gray!10} & \SetCell{gray!10} & \SetCell{gray!10} & \SetCell{gray!10} \\
 13   &  &  &  & 2 & 79 &  2050 & 21204  & \SetCell{gray!10} & \SetCell{gray!10} & \SetCell{gray!10} & \SetCell{gray!10} & \SetCell{gray!10} & \SetCell{gray!10} & \SetCell{gray!10} \\
 14   &  &  &  &  & 46 &  1725 &24359  & \SetCell{gray!10} & \SetCell{gray!10} & \SetCell{gray!10} & \SetCell{gray!10} & \SetCell{gray!10} & \SetCell{gray!10} & \SetCell{gray!10} \\
 15   &  &  &  &  & 25 & 1317 &25604  & \SetCell{gray!10} & \SetCell{gray!10} & \SetCell{gray!10} & \SetCell{gray!10} & \SetCell{gray!10} & \SetCell{gray!10} & \SetCell{gray!10} \\
 16   &  &  &  & 1 & 12 &  951 & 24773  & \SetCell{gray!10} & \SetCell{gray!10} & \SetCell{gray!10} & \SetCell{gray!10} & \SetCell{gray!10} & \SetCell{gray!10} & \SetCell{gray!10} \\
 17   &  &  &  &  & 12 & 630 & 22279  & \SetCell{gray!10} & \SetCell{gray!10} & \SetCell{gray!10} & \SetCell{gray!10} & \SetCell{gray!10} & \SetCell{gray!10} & \SetCell{gray!10} \\
 18   &  &  &  &  & 3 &  416 & 18727  & \SetCell{gray!10} & \SetCell{gray!10} & \SetCell{gray!10} & \SetCell{gray!10} & \SetCell{gray!10} & \SetCell{gray!10} & \SetCell{gray!10} \\
 19   &  &  &  &  & 2 &  245 & 14857  & \SetCell{gray!10} & \SetCell{gray!10} & \SetCell{gray!10} & \SetCell{gray!10} & \SetCell{gray!10} & \SetCell{gray!10} & \SetCell{gray!10} \\
 20   &  &  &  &  & 2 & 143 & 11112  & \SetCell{gray!10} & \SetCell{gray!10} & \SetCell{gray!10} & \SetCell{gray!10} & \SetCell{gray!10} & \SetCell{gray!10} & \SetCell{gray!10} \\
 21   &  &  &  &  & 2 &  81 & 7929  & \SetCell{gray!10} & \SetCell{gray!10} & \SetCell{gray!10} & \SetCell{gray!10} & \SetCell{gray!10} & \SetCell{gray!10} & \SetCell{gray!10} \\
 22   &  &  &  &  & &  56 & 5394  & \SetCell{gray!10} & \SetCell{gray!10} & \SetCell{gray!10} & \SetCell{gray!10} & \SetCell{gray!10} & \SetCell{gray!10} & \SetCell{gray!10} \\
 23   &  &  &  &  & &  28 & 3562  & \SetCell{gray!10} & \SetCell{gray!10} & \SetCell{gray!10} & \SetCell{gray!10} & \SetCell{gray!10} & \SetCell{gray!10} & \SetCell{gray!10} \\
 24   &  &  &  &  & &  13 & 2251  & \SetCell{gray!10} & \SetCell{gray!10} & \SetCell{gray!10} & \SetCell{gray!10} & \SetCell{gray!10} & \SetCell{gray!10} & \SetCell{gray!10} \\
 25   &  &  &  &  & 1 &  11 & 1389  & \SetCell{gray!10} & \SetCell{gray!10} & \SetCell{gray!10} & \SetCell{gray!10} & \SetCell{gray!10} & \SetCell{gray!10} & \SetCell{gray!10} \\
$\sum$ & 1 & 6 & 31 & 203 &1653 & 17156  & 227844\\
\end{tblr}
\caption{The number of arrow-type semigroupoids with a given number of arrows ($\rightarrow$) and objects $(\bullet)$. In other words, these are the number of transitively closed directed graphs with no parallel edges up to isomorphism. Empty cells indicate the impossibility of the graph (number zero), gray cells the currently unknown values (the enumeration is an ongoing computation). The column sums give the number of arrow-type semigroupoids with a given number of objects, similarly, the row sums are for a given number of arrows. The total values for the columns, starting from 6 objects, do contain entries that are not displayed in the table.}
\label{tab-arrow-type-enum}
\end{table}

Enumerating finite categories \cite{2014CountingFiniteCats} depends on the enumeration of finite monoids \cite{smallsemi}, since a category is about connecting the monoids local to objects with cross-object morphisms.
Knowing the large number of monoids even for small number of elements, finding all finite categories is not a particularly rewarding task.
It is limited to small examples.
However, enumerating all arrow-type semigroupoids is a more interesting endeavor: it gives the overall possible connection structures, the `blueprints' of semigroupoids.
Table \ref{tab-arrow-type-enum} contains the known numbers of arrow-type semigroupoids with a given number of arrows and objects.\footnote{Incidentally, the table also has a new question for the ultimate answer forty-two. \emph{How many different transitively closed directed graphs without parallel arrows are there with 4 vertices and 6  arrows?}}

Arrow-type semigroupoids are transitively closed directed graphs with no parallel edges, i.e., only one arrow is allowed for a domain-codomain pair.
They are defined by two parameters: $(\#\text{arrows}, \#\text{objects})$.
This gives some immediate limitations on what is possible.
\begin{itemize}
  \item There cannot be more than $2n$ objects for $n$ arrows without having an isolated object. Allowing isolated objects would imply infinite families for any finite arrow-type semigroupoid.
\item The number of objects limits the number of arrows. For instance, we cannot have more than 4 arrows for 2 objects without introducing parallel arrows. Thus the $(5,2)$ cell is empty in the table.
\item The $(n,2n)$ entries are all 1. This is the case of all distinct domains and codomains, $n$ isolated arrows.
\item The $(n,2n-1)$ entries are all 3. The $(3,3)$ case has 3 possibilities, and then we simply add a separate arrow with two new objects.
\item The $(n^2,n)$ entries are all 1. These are the \emph{complete} arrow-type semigroupoids.
\end{itemize}

\subsection{Implementation: Enumerating Arrow-Type Semigroupoids}

We used three different methods for computing the number of distinct arrow-type semigroupoids with $n$ arrows and $m$ objects:
\begin{description}
\item[brute force] for verification puposes, generating all possibilities for $(n,m)$ combinatorially up to $m^{2n}$ lists and filter them by the required properties of a semigroupoid, 
\item[logic search] for \emph{adding one arrow  while keeping the transitive closure} with or without adding a new object, i.e., computing all $(n,m)$ semigroupoids from $(n-1,m)$ and $(n-1,m-1)$,
\item[transitive closure] systematically \emph{adding one arrow, and compute the generated semigroupoid}, similar to the subsemigroup enumeration method \cite{T4enum,subsemi}.
\end{description}
The brute force method is limited to small cases, but it is an excellent way to bootstrap enumeration by providing ground truth (the algorithm is easier to verify) to check against when developing more specialized algorithms.

Here is the relational code for adding one more arrow. The search space size is small: $m^2$.
The constraints are checking that the arrow has a valid domain and codomain, it is not equal to any existing arrow, and any post- and precomposition gives an arrow already in the set.
Composition is checked for the new arrow as well.
When adding an object too, we have the additional constraint that the new arrow has the new object for its domain, or codomain, or both.
\begin{minted}[bgcolor=LightGray]{clj}
(defn one-more-arrow
  "Finds all the arrows that can be added without violating composition to
   the `arrows` when we have `m` objects available. When `added-object`,
   we need to connect the arrow to the last object."
  [arrows m added-object?]
  (let [[d c] (lvar-vector 2)
        extended (conj arrows [d c])]
    (l/run*
     [q]
     (l/== q [d c])
     (l/membero d (range m))
     (l/membero c (range m))
     (if added-object?
       (l/conda [(l/== (dec m) d)]
                [(l/== (dec m) c)])
       l/succeed)
     (l/distincto extended)
     (l/everyg (fn [[dom cod]]
                 (l/conda ;;postcompose
                  [(l/distincto [cod d])]
                  [(l/membero [dom c] extended)]))
               extended)
     (l/everyg (fn [[dom cod]]
                 (l/conda ;;precompose
                  [(l/distincto [c dom])]
                  [(l/membero [d cod] extended)]))
               extended))))
\end{minted}
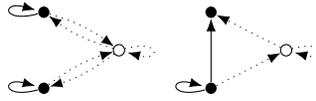
\begin{figure}
  \begin{tikzpicture}
  \tikzset{obj/.style={fill,circle,inner sep=1.5pt},
          arr/.style={-Latex}}
  \node at (0,0) [obj] (t0) {};
  \node at (0,1) [obj] (t1) {};
  \node at (1,0.5) [obj,draw=black, fill=white] (t2) {};

  \path[->,every loop/.append style=arr]
  (t0) edge [arr,loop left]  (t0)
  (t1) edge [arr,loop left] (t1)
  (t2) edge [arr,dotted, bend left=10] (t0)
  (t0) edge [arr,dotted, bend left=10] (t2)
    (t2) edge [arr,dotted, bend left=10] (t1)
  (t1) edge [arr,dotted, bend left=10] (t2)
  (t2) edge [arr,loop right,dotted] (t2);
\end{tikzpicture}
  \begin{tikzpicture}
  \tikzset{obj/.style={fill,circle,inner sep=1.5pt},
          arr/.style={-Latex}}
  \node at (0,0) [obj] (t0) {};
  \node at (0,1) [obj] (t1) {};
  \node at (1,0.5) [obj, draw=black, fill=white] (t2) {};

  \path[->,every loop/.append style=arr]
  (t0) edge [arr,loop left]  (t0)
  (t0) edge [arr] (t1)
  (t0) edge [arr,dotted] (t2)
  (t2) edge [arr,dotted] (t1)
  (t2) edge [arr,dotted, loop right] (t2);
\end{tikzpicture}
\caption{All the ways of adding one arrow  (dotted) and one object (empty circle) to two different $(2,2)$ arrow-type semigroupoids.}
\label{fig:addingone}
\end{figure}
Figure \ref{fig:addingone} visualizes the results of the following two queries.
\begin{minted}[bgcolor=LightGray]{clj}
(one-more-arrow [[0 0] [1 1]] 3 true)
([0 2] [2 0] [1 2] [2 1] [2 2])
(one-more-arrow [[0 1] [1 1]] 3 true)
([0 2] [2 1] [2 2])
\end{minted}
However, this additive method cannot produce all the semigroupoids.
An obvious example is the $(n,2n)$ entries, and the complete semigroupoids that appear after some gap in the table.
One cannot simply add just one arrow when closing cycles, since several appear at the same time.
The most efficient enumeration method is adding one arrow and taking the transitive closure (thus jumping over gaps).
It is similar to the subsemigroup enumeration method that was used for enumerating all degree 4 transformation semigroups \cite{T4enum}.

All methods have the same bottleneck, determining uniqueness.
This is the well-known difficult problem of  graph canonization \cite{MCKAY201494}.
We can build a canonical form (minimal in lexicographic order) by applying all permutations of the objects, up to degree $2n$, the symmetric group $S_{2n}$.
For the brute force approach, this allows a small optimization by starting the underlying lists with zero.
We also need to avoid working with the symmetric groups, since they have $n!$ elements.
Thus, we implemented a logic search for directed graph isomorphism.

Deciding directed graph isomorphism is similar to deciding semigroupoid isomorphism.
The adjacency matrix is sort of like a multiplication table, except that the values are boolean.
Therefore, the search is the same as the semigroupoid homomorphism code, but with a single relation only.
We also added the grouping of possible image arrows by their in- and out-degrees.
This saves checking a lot of doomed to fail assignments.
\begin{minted}[bgcolor=LightGray]{clj}
(defn digraph-isomorphisms
  "Logic search for all isomorphisms of directed graph `G` to `H` given as
   a collection of arrows (ordered pairs of elements of a finite set).
  Limitation: nodes with no arrows are ignored."
  [G H]
  (let
      [Gnodes (nodes G)
       Hdegree-lookup (group-by (out-in-degrees H) (nodes H))
       targets (zipmap Gnodes ;G vertex -> matching H vertices
                       (map (comp Hdegree-lookup
                                  (out-in-degrees G))
                            Gnodes))
       phi (zipmap Gnodes (lvar-vector (count Gnodes))) ;the morphism
       constraints (mapv (fn [[a b]] [(phi a) (phi b)]) G)] ;substitution
    (l/run*
      [q]
      (l/== q phi)
      (l/everyg (fn [v] (l/membero (phi v) (targets v))) Gnodes)
      (l/distincto (vals phi)) ;different vertices go to different vertices
      (l/everyg (fn [arrow]
                  (l/membero arrow (vec H))) ;phi(arrow) should be an arrow in H
                constraints))))
\end{minted}
Using a hash-map to represent the mapping is important.
If the graphs nodes are integers, the search can be easily trapped if we allow directed graphs isolated vertices.
\begin{minted}[bgcolor=LightGray]{clj}
(first (digraph-isomorphisms [[0 0]] [[9 9]]))
\end{minted}
The unification algorithm dutifully matches the isolated vertices in all possible ways, creating a symmetric group.
Therefore, we need to enforce a compact representation, with no isolated vertices.

Despite these optimizations, establishing isomorphism is a heavy computation. It is better to avoid it as much as possible by computing simple invariant properties: e.g., ordered sequences of out- and in-degrees of the graphs, or ordered frequencies of values in the powers of the adjacency graphs (the number of paths of certain lengths).
We implemented a nested hash-map database for isomorphism class representatives, with keys number of nodes, number of arrows, and a graph isomorphism invariant signature.
This allows the enumeration results in Table \ref{tab-arrow-type-enum}, using a computer with 64 gigabyte of RAM memory.

\subsubsection*{Verification} To cross-check the database algorithms we recomputed the number of distinct \emph{functional digraphs} (transformations) for a give degree (number of nodes).
The values of  this integer sequence \cite[\href{https://oeis.org/A001372}{A001372}]{oeis} have been computed in several ways.
There is an efficient algorithm \cite{2024minrependofunc} for computing the class representatives, but for the validation we used the simplest brute force method of listing all degree $m$ transformations and inserting them into the database.
It is a quick computation up to degree 5.

\section{Transformation Representations of Semigroupoids}
\label{sect:transrep}

The underlying topic of this paper is moving between the abstract and a concrete representation.
Here, we want to find a finite transformation for each arrow in a semigroupoid.
Instead of directly constructing the transformations, we will use the algorithms in Section \ref{section:morphisms} and find embeddings into a `full' structure.

For transformation representations, it is natural to ask for the minimal number of states for all types.
First, we need to be able to build the complete semigroupoids from generators.
Then, we characterize the full structures.
Finally, we construct the embeddings into those.

\subsection{Theory: Semigroupoid from Generators}

By a concrete representation we mean attaching some combinatorial objects to the arrows.
Thus, an \emph{arrow} is a triple $(\sigma,\tau,m)$ with source $\sigma$, target $\tau$, and some composable structure $m$.
Two arrows $a=(\sigma,\tau,m)$ and $b=(\sigma',\tau',m')$ are composable if $\tau=\sigma'$.
In that case, the result is $(\sigma,\tau',mm')$.

In a concrete representation, a semigroupoid need not be given in full.
A suitable subset can \emph{generate} the whole semigroupoid by using composition, as the attached combinatorial structure can determine the equality of arrows.
In the semigroup case, we simply keep composing on the right by the generators.
For semigroupoids, we need to compose on both sides, i.e., pre- and postcomposing, but only when types match - using a lookup table for the types of generators.
Objects with no endoarrows and incoming arrows are represented by sets of Garden of Eden states. These have no predecessors, and we can only move away from them.

In a \emph{transformation representation} we assign finite sets to the types, and total functions of these sets to the arrows.
The composition of arrows thus corresponds to function composition.

\begin{example}[Communicating vessels -- Transferring dynamics]
\label{example:communicating_vessels}
  Objects connected by isomorphisms, and thus sets of states with bijective maps between have the same semigroup.
\begin{figure}[h]
  \begin{tikzpicture}[shorten >=1pt, node distance=1.5cm, on grid, auto,inner sep=2pt]
    \tikzset{arr/.style={-Latex},
             rarr/.style={Latex-Latex},
             bg/.style={fill=gray!32,draw=none,rounded corners,inner sep=10pt},
             n/.style={draw,circle,inner sep=1pt,fill=white}}
    \node[n] (X10)   {$0_1$};
    \node[n] (X11) [below of=X10] {$1_1$};
    \node at (3,0) [n] (X20)  {$0_2$};
    \node[n] (X21) [below of=X20] {$1_2$};
    \node[] (X1label) [left = 1cm of X10] {$X_1$};
    \node[] (X2label) [right = 1cm of X20] {$X_2$};
    \path[->,every loop/.append style=arr,every edge/.append style=arr]
      (X10) edge [rarr] node[left] {$c$} (X11)
      (X20) edge node {$r$} (X21)
      (X21) edge [loop below] node (lab) {$r$} (X21)
      (X10.north east) edge [dashed, bend left=10] node {$f_{1\rightarrow 2}$} (X20.north west)
      (X11.north east) edge [dashed, bend left=10] node {$f_{1\rightarrow 2}$} (X21.north west)
      (X20) edge[dotted,thick,bend left=10] node {$g_{2\rightarrow 1}$} (X10)
      (X21) edge[dotted,thick,bend left =10] node {$g_{2\rightarrow 1}$}(X11);
      \begin{pgfonlayer}{background}
        \node[fit=(X20) (X21) (lab), bg] {};
      \end{pgfonlayer}
      \begin{pgfonlayer}{background}
        \node[fit=(X10) (X11), bg] {};
      \end{pgfonlayer}
  \end{tikzpicture}
\caption{Type $X_1$ has a transposition, and type $X_2$ has reset.
The connecting $f,g$ transformations transfer these, so both objects end up with the same permutation-reset automaton.}
\end{figure}
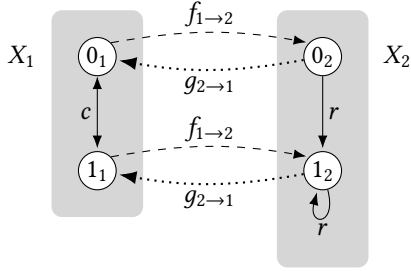
With $c=\left( \begin{smallmatrix}
  0_1 & 1_1 \\ 1_1 & 0_1
\end{smallmatrix}\right)$, $r=\left( \begin{smallmatrix}
  0_2 & 1_2 \\ 1_2 & 1_2
\end{smallmatrix}\right)$,
$f_{1\rightarrow 2}=\left( \begin{smallmatrix}
  0_1 & 1_1 \\ 0_2 & 1_2
\end{smallmatrix}\right)$,
$g_{2\rightarrow 1}=\left( \begin{smallmatrix}
  0_2 & 1_2 \\ 0_1 & 1_1
\end{smallmatrix}\right)$ we can see how the action is transferred.
The cycle goes from $X_1$ to $X_2$ by $g_{2\rightarrow 1}cf_{1\rightarrow 2}=\left( \begin{smallmatrix}
  0_2 & 1_2 \\ 1_2 & 0_2
\end{smallmatrix}\right)$.
The reset goes from $X_2$ to $X_1$ by $f_{1\rightarrow 2}rg_{2\rightarrow 1}=\left( \begin{smallmatrix}
  0_1 & 1_1 \\ 1_1 & 1_1
\end{smallmatrix}\right)$.

Similar transfers explain how bigger groups can be assembled.
For example, if a type realized by a three-element set with a 3-cycle only can generate $S_3$ if it has a bijective map to a two-element set with a transposition.
\end{example}

\subsection{Implementation: \texttt{sgpoid-by-gens}}
We create two lookup tables for the generators.
One maps a type to generators with that type as domain.
This is used for pre-composition.
The other one maps types to generators with that type as a codomain, used for post-composition.
Then, we systematically pre- and postcompose by the generators (starting with the generators themselves) until no new arrow is generated.
Interesting to note that this algorithm has the transitive closure at its core.

\subsection{Theory: From Abstract to Transformation Representations}

Given an abstract semigroupoid (represented as a composition table),
there are at least two ways to construct transformation representations:
\begin{description}
  \item[direct:] For each arrow in the semigroupoid, we find the corresponding transformation. The number of variables is the size of the domain, and each variable can have a number of values that equal to the size of the codomain.
  \item[indirect:] We have a single variable for each arrow and the possible values come from a full transformation semigroupoid. We can use the isomorphism search.
\end{description}
Due to the lower number of logic variables, we choose the second method.
We construct transformation representation by \emph{embedding abstract semigroupoids into full transformation semigroupoids}.
Embedding means finding an isomorphic copy of the source inside the target structure.

The problem is that for transformation semigroupoids, the full structure is not uniquely defined.
For semigroups, we have one \emph{full transformation semigroup} $T_n$ for each \emph{degree} $n$.
It consists of all $n^n$ total transformations of $n$ points.
For semigroupoids `full' in this sense, we can have copies of $T_n$ in each object.
We need to specify the degrees for all of $\Ob(\cS)$.
The local degrees also determine how many arrows can be between to objects.
Then, we need the arrow-type semigroupoid $\cS_\rightarrow$:
the vertices are the objects, and there is an edge between $\tau_i$ and $\tau_j$ iff $\cS$ has an arrow $a:\tau_i\to\tau_j$.
A \emph{full transformation semigroupoid} can be defined relative to a list of degrees $\mathbf{d}$ and a transitively closed directed graph $G$, denoted by $\cT_{\mathbf{d},G}$.
Note that loops $\tau_i\to \tau_i$ might not be present in $G$; thus, a type can have an empty transformation semigroup.

Since the full structure is not unique, we do not have a straightforward way to find the smallest transformation semigroupoids.
For transformation semigroups, we can find the minimal degree representation of $S$ by trying to embed it into $T_1, T_2, T_3$ and so on.
For transformation semigroupoids we need to enumerate transitively closed directed graphs, then find minimal degrees for all objects.
The enumeration of the possible targets is multi-dimensional.
Note that a Cayley Theorem (special case of the Yoneda Lemma) for semigroupoids ensures that there is always a transformation representation, but we aim for the smallest possible ones.

\subsection{Implementation}

The graph is defined by collecting the domain-codomain pairs, conveniently represented by integers $\{0,\ldots,m\}$.
We know minimal generating sets for $T_n$, a cycle, a transposition and an elementary collapsing.
For the type switching maps, we can simply have an arbitrary bijection and let the local $T_n$s scramble them.
However, since we need complete composition tables anyway, we can simply combinatorially compute all possible maps along an arrow.

\begin{example}
  We construct transformation representations for the semigroupoid in Example \ref{ex:2obj-sgpoid} (Figure \ref{2obj6arr}).
  First we construct a few candidate full transformation semigroupoids (see Figure \ref{fig:target_graphs}) for the embeddings.
  Using degree 2 for the objects, the embedding into the fully connected and the one-way connected two objects give the same representation.
  The fully connected giving more solutions due to the extra symmetry of swapping objects.
  The isolated objects cannot emulate the original semigroupoid.
  However, there is a (non-strict) embedding into $T_3$.
  The image of the 6-arrow semigroupoid is a 7-element transformation semigroup, as the non-composable arrows map to a sink state.
\end{example}

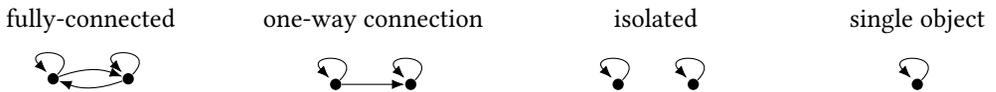
\begin{figure}[h]
  \centering
\begin{tblr}{Q[c,m,.24\linewidth]Q[c,m,.24\linewidth]Q[c,m,.24\linewidth]Q[c,m,.2\linewidth]}
fully-connected&  one-way connection& isolated & single object\\
  \begin{tikzpicture}
  \tikzset{obj/.style={fill,circle,inner sep=1.5pt},
          arr/.style={-Latex}}
  \node at (0,0) [obj] (t0) {};
  \node at (1,0) [obj] (t1) {};

  \path[->,every loop/.append style=arr]
  (t0) edge [arr,loop]  (t0)
  (t1) edge [arr,loop] (t1)
  (t0) edge [arr, bend left=20] (t1)
  (t1) edge [arr, bend left=20] (t0);
\end{tikzpicture}&
  \begin{tikzpicture}
  \tikzset{obj/.style={fill,circle,inner sep=1.5pt},
          arr/.style={-Latex}}
  \node at (0,0) [obj] (t0) {};
  \node at (1,0) [obj] (t1) {};

  \path[->,every loop/.append style=arr]
  (t0) edge [arr,loop]  (t0)
  (t1) edge [arr,loop] (t1)
  (t0) edge [arr] (t1);
\end{tikzpicture}&
\begin{tikzpicture}
  \tikzset{obj/.style={fill,circle,inner sep=1.5pt},
          arr/.style={-Latex}}
  \node at (0,0) [obj] (t0) {};
  \node at (1,0) [obj] (t1) {};

  \path[->,every loop/.append style=arr]
  (t0) edge [arr,loop]  (t0)
  (t1) edge [arr,loop] (t1);
\end{tikzpicture}&
\begin{tikzpicture}
  \tikzset{obj/.style={fill,circle,inner sep=1.5pt},
          arr/.style={-Latex}}
  \node at (0,0) [obj] (t0) {};
  \path[->,every loop/.append style=arr]
  (t0) edge [arr,loop] (t0);
\end{tikzpicture}
  \end{tblr}
  \caption{Possible target graphs (arrow-type semigroupoids) for full transformation semigroupoids. When choosing degrees $\mathbf{d}=(2,2)$ for the first three graphs, every arrow in these diagrams stand for the 4 possible transformations on 2 points. For the single object representation, we need degree at least 3, to have 27 possible candidate transformations, since 4 would not be enough.}
  \label{fig:target_graphs}
\end{figure}

\section{Conclusion and Further Work}
\label{sect:conclusions}

We explored associativity and type structure consistency for abstract semigroupoids and created transformation representations by embedding into certain full transformation semigroupoids.
We succeeded improving our understanding of semigroupoids as models of computation in several ways.
\begin{itemize}
  \item There can be different type structures for the same abstract semigroupoid.
  \item We can define homomorphisms in a strict, typed sense, or a more permissive embedding style, but leading to some weaker properties (bijective homomorphisms not being isomorphisms, embedding into composition tables failing associativity).
  \item The arrow-type semigroupoid is always a valid homomorphic image and it serves as a blueprint for all semigroupoids, thus we started the enumeration of the finite ones.
\end{itemize}
These results are, of course, only the beginning of investigating these structures.
The most immediate todo items for further exploration are the following.
\begin{itemize}
  \item Unify the associativity and type consistency constraints in a single enumeration.
  \item Construct relation morphisms.
\end{itemize}
We also have a wishlist for the logic engine.
\begin{itemize}
  \item The goal \texttt{geto} for hash-maps would be useful.
In functional mode in \texttt{Clojure} we frequently use hash-maps.
When going relational, we need to convert the algorithms back to list-based operations.
\item How to do existential constraints efficiently? E.g., $\exists a,b,c,d,e,f$ such that $abc\neq def$ for excluding 3-nilpotency.
\end{itemize}

\section*{Software Tools}

All the algorithms described in this paper are implemented in the \texttt{kigen} open-source software package \cite{kigen}.
The subfolder \texttt{2025\_Computational\_Exploration\_of\_Finite\_Semigroupoids} containing the computations for this paper can be found in the \texttt{experiments} folder of the package.

\begin{acks}
  We thank Jason Hemann and William E.~Byrd for their inspiration and for writing this paper.
  We also thank two anonymous reviewers for their constructive criticism.
This project was funded in part by the Kakenhi grant
22K00015 by the Japan Society for the Promotion of Science (JSPS), titled `On progressing human understanding in the shadow of superhuman
deep learning artificial intelligence entities' (Grant-in-Aid for Scientific
Research type C, \url{https://kaken.nii.ac.jp/grant/KAKENHI-PROJECT-22K00015/}). This research was also supported in part by the Natural Sciences and Engineering Research Council of Canada (NSERC), funding reference number RGPIN-2019-04669.
Cette recherche a \'et\'e financ\'ee en partie par le Conseil de recherches en sciences naturelles et en g\'enie du Canada (CRSNG), num\'ero de r\'ef\'erence RGPIN-2019-04669.
\end{acks}

\bibliographystyle{ACM-Reference-Format}
\bibliography{../coords.bib}

\end{document}